\documentclass[12pt]{myiopart}
\usepackage{amsmath, amssymb, amsfonts, amsthm}
\usepackage{hyperref}

\newcommand{\x}{\mathbf{x}}

\newcommand{\y}{\mathbf{y}}

\newcommand{\al}{\boldsymbol{\alpha}}

\newcommand{\rxy}{|\x- \y|}

\newcommand{\rhs}{$\mathrm{r.\,h.\,s.}$ }
\newcommand{\lhs}{$\mathrm{l.\,h.\,s.}$ }

\newtheorem{theorem}{Theorem}[section]
\newtheorem{lemma}[theorem]{Lemma}
\newtheorem{proposition}[theorem]{Proposition}
\newtheorem{remark}[theorem]{Remark}
\newtheorem{assumption}[theorem]{Property}
\newtheorem{corollary}[theorem]{Corollary}

\DeclareMathOperator{\esssup}{ess~sup}
\DeclareMathOperator{\Spec}{Spec}

\makeatletter
\@addtoreset{equation}{section}
\makeatother

\eqnobysec

\begin{document}

\title{Exponential decay of eigenfunctions of Brown--Ravenhall operators}
\author{Sergey Morozov}
\address{Department of Mathematics, University College London, Gower Street, London, WC1E 6BT, UK}
\ead{morozov@math.ucl.ac.uk}

\begin{abstract}
We prove the exponential decay of eigenfunctions of reductions of Brown--Ravenhall operators to arbitrary irreducible representations of rotation--reflection and permutation symmetry groups under the assumption that the corresponding eigenvalues are below the essential spectrum.
\end{abstract}

\section{Introduction}

The Brown--Ravenhall operator can be considered as the (multiparticle) Dirac operator projected to the positive spectral subspace of free particles.
This operator was introduced in \cite{BrownRavenhall1951} as a Hamiltonian of quantum electrodynamics (QED) correct to the second order in the fine structure constant (see also \cite{Sucher1980}).
The higher order corrections predicted by QED should thus be treated as perturbations. The Brown--Ravenhall model turns out to be a good candidate for this approach, as the recent rigorous results show. Indeed, it is bounded below even in the many--particle case for physically relevant nuclear charges \cite{EvansPerrySiedentop1996, BalinskyEvans1999, HoeverSiedentop1999}, and the structure of its spectrum resembles the one of Schr\"odinger operator -- the essential spectrum forms a semiaxis \cite{MorozovVugalter2006, Jakubassa2005, Jakubassa2007, Morozov2008}, possibly with some eigenvalues below ionization thresholds \cite{MorozovVugalter2006, Morozov2008}. This is in a remarkable contrast to the many--particle Coulomb--Dirac operator which has essential spectrum on the whole real axis and no eigenvalues, but is sometimes used as a formal unperturbed Hamiltonian in some QED calculations.

Having in mind the intention to consider the Brown--Ravenhall operator as an unperturbed intermediate model, it is very useful to have information on the rate of spatial decay of its eigenfunctions. In this article we prove that for systems of particles with electric charges of the same sign (we consider the potential energy of interactions with nuclei as external field) the eigenfunctions decay exponentially provided the corresponding eigenvalues are below the essential spectrum. This will also be proved for restrictions of the operator to subspaces of wavefunctions with certain rotation--reflection symmetries.

There are numerous results concerning the exponential decay of eigenfunctions of multiparticle Schr\"odinger operators, including anisotropic estimates and lower bounds. A very detailed analysis of the non--isotropic exponential decay of eigenfunctions of Schr\"odinger operators in terms of a metric in configuration space is presented in \cite{Agmon1982}. It is proved in \cite{CarmonaSimon1981} that the upper bound of \cite{Agmon1982} is exact at least for the ground state. A very simple proof of the exponential decay, based on the approach of \cite{Agmon1982} can be found in \cite{GriesemerLiebLoss2001}, Lemma 6.2. 

As for relativistic operators, exponential decay of eigenfunctions is proved for one--particle Chandrasekhar operators \cite{Nardini1986, CarmonaMastersSimon1990} and some projected multiparticle Dirac operators \cite{MatteStockmeyer2008preprint}. For one--particle Brown--Ravenhall atomic Hamiltonians the exponential decay of eigenfunctions was first obtained in \cite{BachMatte2001} for coupling constants of Coulomb potential not exceeding $\frac12$. In the recent preprint \cite{MatteStockmeyer2008preprint2} the exponential decay of bigger rate is shown to hold pointwise for all one--electron atoms with subcritical or critical coupling constants.

The paper is organized as follows.
In Section~\ref{model section} we introduce the Brown--Ravenhall model together with some auxiliary constructions and formulate the main result in Theorem~\ref{decay theorem}.
Then in Section~\ref{modified assumptions section} we discuss the relevant properties of the interaction potentials. The proof of Theorem~\ref{decay theorem} is presented in Section~\ref{proof of decay theorem section}, with the proofs of technical lemmata postponed until Sections~\ref{proof of the main lemma}---\ref{Q_4 proof section}. In Section~\ref{integral operators section} we prepare these proofs recalling two useful theorems which give sufficient conditions for the boundedness of integral operators. Appendix contains a couple of properties of modified Bessel functions for reference.

\section{The model and the main result}\label{model section}

In the Hilbert space $L_2(\mathbb{R}^3, \mathbb{C}^4)$ the Dirac operator describing a particle of mass $m> 0$ is given by
\begin{equation*}\label{Dirac}
D_m= -i\al\cdot\nabla+ \beta m,
\end{equation*}
where $\al:= (\alpha_1, \alpha_2, \alpha_3)$ and $\beta$ are the $4\times 4$ Dirac matrices \cite{Thaller1992}.
The form domain of $D_m$ is the Sobolev space $H^{1/2}(\mathbb{R}^3, \mathbb{C}^4)$ and its spectrum is $(-\infty, -m]\cup[m, +\infty)$. Let $\Lambda_m$ be the orthogonal projector onto the positive spectral subspace of $D_m$:
\begin{equation*}\label{Lambda_m}
\Lambda_m:=\frac{1}{2}+ \frac{-i\al\cdot\nabla+ \beta m}{2\sqrt{-\Delta+ m^2}}.
\end{equation*}
We consider a finite system of $N$ particles with positive masses $m_n, \: n= 1, \dots, N$. To simplify the notation we write $D_n$ and $\Lambda_n$ for $D_{m_n}$ and $\Lambda_{m_n}$, and also for their tensor products with the identity operators in $L_2(\mathbb R^3, \mathbb C^4)$, e. g.
\begin{equation*}
\underset{j= 1}{\overset{n- 1}\otimes}I\otimes D_{m_n}\otimes\underset{k= n+ 1}{\overset{N}\otimes}I \quad \textrm{and}\quad \underset{j= 1}{\overset{n- 1}\otimes}I\otimes \Lambda_{m_n}\otimes\underset{k= n+ 1}{\overset{N}\otimes}I,
\end{equation*}
respectively.

Let $\mathfrak{H}_N:= \underset{n= 1}{\overset{N}{\otimes}}\Lambda_nL_2(\mathbb{R}^3, \mathbb{C}^4)$ be the Hilbert space with the inner product induced by the one on $\underset{n= 1}{\overset{N}{\otimes}}L_2(\mathbb{R}^3, \mathbb{C}^4)\cong L_2(\mathbb{R}^{3N}, \mathbb{C}^{4^N})$.
In this space the $N$--particle Brown--Ravenhall operator is formally defined by
\begin{equation}\label{H_N'}
\mathcal{H}_N= \Lambda^N\bigg(\underset{n= 1}{\overset{N}{\sum}}(D_n+ V_n)+ \underset{n< j}{\overset{N}{\sum}}U_{nj}\bigg)\Lambda^N,
\end{equation}
with
\begin{equation*}\label{Lambda^N}
\Lambda^N:= \prod_{n= 1}^N\Lambda_n= \underset{n= 1}{\overset N\otimes}\Lambda_n.
\end{equation*}
Here the indices $n$ and $j$ indicate the particle on whose coordinates the corresponding operator acts.

In \eqref{H_N'} $V_n$ and $U_{nj}$ are the operators of multiplication by the potential energy of interactions of the particles of the system with an external field and between themselves, respectively.
In most applications to atomic and molecular physics Brown--Ravenhall operators are considered in the Born--Oppenheimer approximation. Then $V_n$ is the potential energy of the $n^{th}$ particle in the electrostatic field of static nuclei
\begin{equation}\label{V_n}
V_n(\x_n):= e_n\sum_{k= 1}^K\frac{z_k}{|\x_n- \mathbf r_k|},
\end{equation}
where $e_n$ is the electric charge of the particle, and $z_k$ and $\mathbf r_k$ are the charges and positions of the nuclei. The interaction between the particles is given by the Coulomb potential energy
\begin{equation}\label{U_nj}
U_{nj}(\x_n, \x_j):= \frac{e_ne_j}{|\x_n- \x_j|}.
\end{equation}
We will assume that all the particles of the system have the same sign of electric charges $e_n$, $n= 1, \dots, N$, but otherwise they might be different, as happens for \emph{exotic atoms}, where some electrons are replaced with muons or even hadrons. The spin of each particle is assumed to be equal to 1/2, as always with Dirac and Brown--Ravenhall operators. This implies that the particles of the system are fermions. According to the Pauli principle, if some of the particles are identical, the wavefunction of the system should be antisymmetric under their permutations. This means that the operator \eqref{H_N'} should be restricted to the subspace of $\mathfrak{H}_N$ consisting of functions which transform according to a certain irreducible representation $E$ of a subgroup $\Pi$ of the symmetric group $\mathcal S_N$ generated by transpositions of identical particles. Let $P^E$ be the orthogonal projector in $\mathfrak{H}_N$ onto the space of such functions. We will denote the restriction of $\mathcal H_N$ to $\mathfrak H^E:= P^E\mathfrak H$ by $\mathcal H_N^E$.

We will assume that the subcriticality condition
\begin{equation}\label{subcriticality}
\min_{n, k}e_nz_k>- 2(2/\pi+ \pi/2)^{-1}
\end{equation}
holds. According to \cite{BalinskyEvans1999}, $\mathcal{H}_N$ (and thus $\mathcal H_N^E$) is bounded below even if we replace the strict inequality in \eqref{subcriticality} by a non--strict. Violation of such non--strict inequality usually leads to the lack of boundedness below, as shown in \cite{EvansPerrySiedentop1996} for the case of single nucleus. As far as $\mathcal{H}_N$ (or any of its restrictions) is bounded below, it can be defined via the corresponding quadratic form.

It is convenient to reduce $\mathcal H_N^E$ using the rotation--reflection symmetries of the system.
Let $\gamma$ be an orthogonal transform in $\mathbb{R}^3$: the rotation around the axis directed along a unit vector $\mathbf{n}_\gamma$ through an angle $\varphi_\gamma$, possibly combined with the reflection $\x\mapsto -\x$. The corresponding unitary operator $O_\gamma$ acts on the functions $\psi\in\mathfrak{H}_N$ as (see \cite{Thaller1992}, Chapter 2)
\begin{equation*}\label{rotation}
(O_\gamma\psi)(\x_1, \dots, \x_N)= \prod_{n= 1}^Ne^{-i\varphi_\gamma\mathbf{n}_\gamma\cdot\mathbf{S}_n}\psi(\gamma^{-1}\x_1, \dots, \gamma^{-1}\x_N).
\end{equation*}
Here $\mathbf{S}_n= - \frac i4\alpha_n\wedge\alpha_n$ is the spin operator acting on the spinor coordinates of the $n^{th}$ particle.
The compact group of orthogonal transformations $\gamma$ such that $O_\gamma$ commutes with $V_n$ and $U_{nj}$ for all $n,\,j= 1, \dots, N$ (and thus with $\mathcal{H}_N^E$) we denote by $\Gamma$. Further, we decompose $\mathfrak{H}_N^E$ into the orthogonal sum
\begin{equation}\label{representations}
\mathfrak{H}_N^E= \underset{T\in \mathrm{Irr}\,\Gamma}{\oplus}\mathfrak{H}_N^{T, E},
\end{equation}
where $\mathfrak{H}_N^{T, E}$ consists of functions form $\mathfrak H_N^E$ which transform under $O_\gamma$ according to some irreducible representation $T$ of $\Gamma$. The decomposition \eqref{representations} reduces $\mathcal{H}_N^E$. We denote the selfadjoint restrictions of $\mathcal{H}_N^E$ to $\mathfrak{H}_N^{T, E}$ by $\mathcal{H}_N^{T, E}$. The spectrum of $\mathcal{H}_N^E$ is the union of the spectra of $\mathfrak{H}_N^{T, E}$, $T\in \mathrm{Irr}\,\Gamma$.

Together with the whole system of $N$ particles we will consider its decompositions into two clusters. Such decompositions play an important role in the characterization of the essential spectrum of the operators $\mathfrak{H}_N^{T, E}$.
Let $Z= (Z_1, Z_2)$ be a decomposition of the index set $I:= \{1, \dots, N\}$ into two disjoint subsets:
\begin{equation*}\label{decomposition}
I= Z_1\cup Z_2, \quad Z_1\cap Z_2= \varnothing.
\end{equation*}
Let
\begin{equation}\label{H_Z_1}
\widetilde{\mathcal{H}}_{Z, 1}:= \sum_{n\in Z_1}(D_n+ V_n)+ \sum_{\substack{n,j\in Z_1\\ n< j}}U_{nj},
\end{equation}
\begin{equation}\label{H_Z_2}
\widetilde{\mathcal{H}}_{Z, 2}:= \sum_{n\in Z_2}D_n+ \sum_{\substack{n,j\in Z_2\\ n< j}}U_{nj}.
\end{equation}
We introduce the operators corresponding to noninteracting clusters, with the second cluster transferred far away from the sources of the external field:
\begin{equation}\label{reduced spaces}
\mathcal{H}_{Z, j}:= \Lambda_{Z, j}\widetilde{\mathcal{H}}_{Z, j}\Lambda_{Z, j},\quad \textrm{in}\quad \mathfrak{H}_{Z, j}:= \underset{n\in Z_j}{\otimes}\Lambda_nL_2(\mathbb{R}^3, \mathbb{C}^4), \quad j= 1, 2,
\end{equation}
where
\begin{equation*}\label{Lambda_Zj}
\Lambda_{Z, j}:= \underset{n\in Z_j}{\prod}\Lambda_n= \underset{n\in Z_j}{\otimes}\Lambda_n.
\end{equation*}

For a given cluster decomposition $Z= (Z_1, Z_2)$ we denote by $P^{E_j}$ and $P^{T_j}$ the projectors onto the irreducible representations $E_j$ and $T_j$ of the restrictions of $\Pi$ and $\Gamma$, respectively, to the cluster of particles indexed by $Z_j$, $j= 1, 2$.

Given representations $T_j$ and $E_j$, projector $P^{T_j}P^{E_j}= P^{E_j}P^{T_j}$ reduces $\mathcal{H}_{Z, j}$. We denote the reduced operators in 
\begin{equation*}\label{H^T_j^E_j}
\mathfrak{H}_{Z, j}^{T_j, E_j}:= P^{T_j}P^{E_j}\mathfrak{H}_{Z, j}
\end{equation*}
by $\mathcal{H}_{Z, j}^{T_j, E_j}$, and define
\begin{equation}\label{E}
\varkappa_j(Z, T_j, E_j):= \inf \Spec\mathcal{H}_{Z, j}^{T_j, E_j}.
\end{equation}

We write $(T_1, E_1; T_2, E_2)\underset Z\prec (T, E)$ if the corresponding term cannot be omitted on the \rhs of
\begin{equation*}\label{symmetry embedding 2}
\mathfrak{H}_N^{T, E}\subset \underset{\substack{(T_1, E_1)\\ (T_2, E_2)}}{\oplus}\big(\mathfrak{H}_{Z, 1}^{T_1, E_1}\otimes\mathfrak{H}_{Z, 2}^{T_2, E_2}\big)
\end{equation*}
without violation of the inclusion.
For $Z_2\neq \varnothing$ let
\begin{equation}\label{bottom}\begin{split}
&\varkappa(Z, T, E):= \begin{cases}\underset{(T_1, E_1; T_2, E_2)\underset Z\prec (T, E)}\inf\big\{\varkappa_1(Z, T_1, E_1)+ \varkappa_2(Z, T_2, E_2)\big\}, &Z_1\neq \varnothing,\\ \varkappa_2(Z, T, E), &Z_1= \varnothing,\end{cases}
\end{split}\end{equation}
and
\begin{equation}\label{E(T)}
\varkappa(T, E):= \min\big\{\varkappa(Z, T, E): Z=(Z_1, Z_2), \: Z_2\neq \varnothing\big\}.
\end{equation}

We are now ready to characterize the essential spectrum of $\mathcal{H}_N^{T, E}$ in terms of cluster decompositions:
\begin{theorem}{\bfseries\em (Morozov \cite{Morozov2008}, Theorem~6)}\label{essential theorem}
For $N\in\mathbb N$ let $T$ be some irreducible representation of $\Gamma$, and $E$ some irreducible representation of $\Pi$, such that $P^TP^E\neq 0$.
The essential spectrum of $\mathcal{H}_N^{T, E}$ is $\big[\varkappa(T, E), \infty\big)$.
\end{theorem}
Thus the bottom of the essential spectrum is equal to the minimal energy which the system can have if some of the particles are transfered far away form other particles and sources of external field. We will omit the proof of the following simple proposition based on the positivity of the interaction potentials \eqref{U_nj}.

\begin{proposition}
It is enough to take the minimum in \eqref{E(T)} over $Z$ with $Z_2= \{n\}$, $n= 1, \dots N$. Moreover, for such $Z$, $\varkappa_2(Z, \cdot, \cdot)$ in \eqref{bottom} is equal to $m_n$, the mass of the particle in the second cluster. 
\end{proposition}

As shown in \cite{MorozovVugalter2006}, the Brown--Ravenhall operators $\mathcal{H}_N^{T, E}$ can have eigenvalues below the essential spectrum. Note that in view of the decomposition \eqref{representations} these eigenvalues can be embedded in the essential spectrum of $\mathfrak{H}_N^E$.

Our main result is the following theorem.

\begin{theorem}\label{decay theorem}
For $N\in\mathbb N$ let $T$ be some irreducible representation of $\Gamma$, and $E$ some irreducible representation of $\Pi$, such that $P^TP^E\neq 0$.
Let $\phi$ be an eigenfunction of $\mathcal{H}_N^{T, E}$ corresponding to an eigenvalue $\lambda$ below the essential spectrum, i.e.
\begin{equation*}\label{eigenvalue equation}
\mathcal{H}_N^{T, E}\phi= \lambda\phi, \qquad \lambda< \varkappa(T, E).
\end{equation*}
Then there exists $S> 0$ independent of $\lambda$ and $\phi$ such that for
\begin{equation*}
s:= \min\Big\{\frac1{2\sqrt N}, \big(\varkappa(T, E)- \lambda\big)S\Big\}
\end{equation*}
it holds
\begin{equation}\label{decay statement}
\int_{\mathbb R^{3N}} \rme^{2s|\mathbf{X}|}\big|\phi(\mathbf{X})\big|^2\rmd\mathbf{X}< \infty.
\end{equation}
\end{theorem}

Note that for $\lambda$ close to the bottom of the essential spectrum $s$ behaves linearly in $\big(\varkappa(T, E)- \lambda\big)$. However, for Schr\"odinger \cite{Agmon1982}, Dirac \cite{HelfferParisse1994}, Chandrasekhar \cite{CarmonaMastersSimon1990}, and one--particle Brown--Ravenhall operators \cite{MatteStockmeyer2008preprint2} $s$ can be chosen to be proportional to the square root of this distance. This suggests a conjecture that for the multiparticle operators we are considering the actual rate of decay might have this property as well. But the proof of such a conjecture is yet obscure even in view of \cite{MatteStockmeyer2008preprint2}, since that result is obtained by comparison to the decay rate of the eigenfunctions of Dirac operator, which are nonexistent in the multiparticle case.

\section{Some properties of the model}\label{modified assumptions section}

In this section we single out some simple properties of the multiparticle Brown--Ravenhall operators introduced in the previous section. The reason for doing so is twofold. First, it will allow the reader to see which properties are required in each step of the subsequent proof of the exponential decay. Second, this will allow us to reformulate the main result without referring to the explicit form of the potentials \eqref{V_n} and \eqref{U_nj}, thus making future generalizations easier.

We need a bit of notation.
Let $\{\Omega_j\}_{j= 1}^N$ be a collection of uniformly $C^1$-regular domains in $\mathbb R^3$ with bounded boundaries.
For $n= 1, \dots, N$, $s\in\mathbb R$, and $\Omega= \underset{j= 1}{\overset N\times}\Omega_j$ we introduce the anisotropic Sobolev spaces
\begin{equation*}\label{H^s_n}
H^s_n(\Omega, \mathbb C^{4^N}):= \big(\underset{j= 1}{\overset{n- 1}\otimes} L_2(\Omega_j, \mathbb C^4)\big)\otimes H^s(\Omega_n, \mathbb C^4)\otimes\big(\underset{j= n+ 1}{\overset N\otimes} L_2(\Omega_j, \mathbb C^4)\big).
\end{equation*}

\begin{assumption}\label{square integrability assumption old}
For any $R> 0$ there exists a finite $C_R\geqslant 0$ such that
\begin{equation*}\label{V and U integrability again}
\sum_{n= 1}^N\bigg(\int_{|\x|\leqslant R}\big|V_n(\x)\big|^2\rmd\x\bigg)^{1/2}+ \sum_{n< j}^N\bigg(\int_{|\x|\leqslant R}\big| U_{nj}(\x)\big|^2\rmd\x\bigg)^{1/2}\leqslant C_R.
\end{equation*}
\end{assumption}
\noindent In other words, the interaction potentials are locally square integrable.

\begin{assumption}\label{decay assumption again}
The external field potentials decay at infinity in the $L_\infty$--norm:
\begin{equation}\label{V_n decay again}
\lim_{R\to\infty}\underset{|\x|> R}\esssup\big|V_n(\x)\big|= 0, \quad n= 1, \dots, N.
\end{equation}
\end{assumption}

\begin{assumption}\label{U_nj old decay assumption}
For any $\varepsilon> 0$ there exists $R> 0$ big enough such that for all $n< j= 1, \dots, N$
\begin{equation*}\label{U_nj old decay}
\|U_{nj}\psi\|_{L_2(\mathbb R^{3N}\cap\{|\x_n- \x_j|> R\})}\leqslant \varepsilon\min_{k= n, j}\|\psi\|_{H^{1/2}_k(\mathbb R^{3N}, \mathbb{C}^{4^N})}, \quad \textrm{for all}\quad \psi\in H^{1/2}(\mathbb{R}^{3N}, \mathbb{C}^{4^N}).
\end{equation*}
\end{assumption}
\begin{proof}
This follows from the weaker property
\begin{equation*}\label{U_nj decay again}
\lim_{R\to\infty}\underset{|\x_n- \x_j|> R}\esssup\big|U_{nj}(\x)\big|= 0, \quad n, j= 1, \dots, N
\end{equation*}
of the potentials \eqref{U_nj}.
\end{proof}

\begin{assumption}\label{U positivity assumption}
The interparticle interaction potentials are nonnegative:
\begin{equation}\label{U positivity}
U_{nj}\geqslant 0, \quad \textrm{for all}\quad n< j= 1, \dots, N.
\end{equation}
\end{assumption}
\noindent This follows from the assumption that all the particles of the system have electric charges of the same sign.

\begin{assumption}\label{upper assumption old}
There exists $C>  0$ such that for any $n= 1, \dots, N$
\begin{equation}\label{upper estimate V}
\big|\langle V_n\varphi, \psi\rangle\big|\leqslant C\|\varphi\|_{H_n^{1/2}}\|\psi\|_{H_n^{1/2}},\quad \textrm{for any}\quad \varphi, \psi\in H_n^{1/2}(\mathbb{R}^{3N}, \mathbb{C}^{4^N}),
\end{equation}
and for any $n< j= 1, \dots, N$
\begin{equation}\label{upper estimate U}
\big|\langle U_{nj}\varphi, \psi\rangle\big|\leqslant C\|\varphi\|_{H^{1/2}}\|\psi\|_{H^{1/2}},\quad \textrm{for any}\quad \varphi, \psi\in H^{1/2}(\mathbb{R}^{3N}, \mathbb{C}^{4^N}).
\end{equation}
\end{assumption}
\begin{proof}{}
Inequalities \eqref{upper estimate V} and \eqref{upper estimate U} follow from Kato's inequality (see \cite{Kato1966} and \cite{Herbst1977}, Theorem 2.9a.)
\end{proof}

\begin{assumption}\label{U Hardy assumption}
There exists $C> 0$ such that for any $n= 1, \dots, N$ and any $\psi\in H^1(\mathbb{R}^{3N}, \mathbb{C}^{4^N})$
\begin{equation}\label{U Hardy}
\|U_{nj}\psi\|\leqslant C\min_{k= n, j}\|\psi\|_{H^1_k(\mathbb{R}^{3N}, \mathbb{C}^{4^N})}.
\end{equation}
\end{assumption}
\noindent It is not surprising to have the minimum on the \rhs of \eqref{U Hardy}, since $U_{nj}$ only depends on the difference $\x_n- \x_j$. Note that \eqref{U Hardy} can be applied even if $\psi$ is only known to belong either to $H^1_n(\mathbb{R}^{3N}, \mathbb{C}^{4^N})$ or to $H^1_j(\mathbb{R}^{3N}, \mathbb{C}^{4^N})$.
\begin{proof}{}
Inequality \eqref{U Hardy} follows from Hardy's inequality (see e.g. \cite{BirmanSolomjak1987}, page 55) and the properties of symmetric--decreasing rearrangements (see e.g. \cite{LiebLoss2001}, Lemma~7.17 and relation (3.3.4)).
\end{proof}

\begin{assumption}\label{lower assumption old}
There exist $C_1> 0$ and $C_2\in \mathbb{R}$ such that for any cluster decomposition $Z$
\begin{equation*}\label{assumption 1 again}\begin{split}
\langle\mathcal{H}_{Z, j}\psi, \psi\rangle\geqslant C_1\langle\sum_{n\in Z_j}D_n\psi, \psi\rangle- C_2\|\psi\|^2&, \\ \textrm{for any}\quad \psi\in\underset{n\in Z_j}{\otimes}\Lambda_nH^{1/2}(\mathbb{R}^3, \mathbb{C}^4)&, \quad j= 1, 2.
\end{split}\end{equation*}
\end{assumption}
\begin{proof}{}
This is where we need the subcriticality condition \eqref{subcriticality}. According to the result of \cite{BalinskyEvans1999}, inequality \eqref{subcriticality} implies \eqref{assumption 1 again} if $N= 1$. For $N> 1$ it is enough to use Proprety~\ref{U positivity assumption} to estimate $\mathcal H_{Z, j}$ from below by a direct sum of one--particle operators.
\end{proof}

\begin{remark}
By Properties~\ref{upper assumption old} and \ref{lower assumption old}, the quadratic forms of operators \eqref{reduced spaces} (and, in particular, $\mathcal H_N$) are bounded below and closed on $\underset{n\in Z_j}{\otimes}\Lambda_nH^{1/2}(\mathbb{R}^3, \mathbb{C}^4)$. Thus these operators are well--defined in the form sense, and so are their restrictions to invariant subspaces.
\end{remark}

\begin{remark}\label{generalized decay theorem}
Suppose that the potentials \eqref{V_n} and \eqref{U_nj} are replaced by operators of multiplication by some measurable hermitian matrix--valued functions such that $V_n$ are the operators of multiplication of spinor coordinates of $n^{th}$ particle by $4\times 4$ matrix--valued functions $V_n(\x_n)$, $n= 1, \dots, N$, and $U_{nj}$ are the operators of multiplication of spinor coordinates of $n^{th}$ and $j^{th}$ particles by $16\times 16$ matrix--valued functions $U_{nj}(\x_n- \x_j)$, $n< j= 1, \dots, N$.
Then the statements of Theorems~\ref{essential theorem} and \ref{decay theorem} remain valid provided Properties~\ref{square integrability assumption old} --- \ref{lower assumption old} hold. Indeed, Properties~\ref{square integrability assumption old}--- \ref{lower assumption old} imply Assumptions~1---5 of \cite{Morozov2008}, which form the hypothesis of Theorem 6 of \cite{Morozov2008}. And in the proof of Theorem~\ref{decay theorem} we will not need the explicit expressions \eqref{V_n} and \eqref{U_nj}, but only the properties listed in this section.
\end{remark}

\section{Proof of Theorem~\ref{decay theorem}}\label{proof of decay theorem section}

Some constants in the proof can depend on the masses of the particles. Since we only deal with a finite number of particles with positive masses, such dependence will not be indicated explicitly.

\begin{lemma}\label{s to a lemma}
Suppose that for some $a> 0$
\begin{equation}\label{simpler statement}
\int_{\mathbb R^{3N}} \rme^{2a|\x_n|}\big|\phi(\mathbf{X})\big|^2\rmd\mathbf{X}< \infty, \quad n= 1, \dots, N.
\end{equation}
Then \eqref{decay statement} holds with $s= N^{-1/2}a$.
\end{lemma}
\begin{proof}{} 
\begin{equation*}\label{exponent game}
\rme^{2s|\mathbf{X}|}\leqslant \rme^{2\sqrt{N}s\underset{n= 1, \dots, N}{\max}|\x_n|}\leqslant \sum_{n= 1}^N\rme^{2\sqrt{N}s|\x_n|}= \sum_{n= 1}^N\rme^{2a|\x_n|}.
\end{equation*}
Thus \eqref{simpler statement} implies \eqref{decay statement} after summation in $n$.
\end{proof}

It remains to prove that \eqref{simpler statement} holds with some suitable $a> 0$.
Without loss of generality we will consider the case $n= 1$.

Let $\rho\in C^2\big([0, \infty), [0, \infty)\big)$ be given by
\begin{equation}\label{rho}
\rho(z):=\begin{cases}z^2- \dfrac{z^3}{3}, & z\in[0,1),\\ \\z- \dfrac{1}{3}, & z\in[1, \infty).\end{cases}
\end{equation}
For $\epsilon> 0$ let
\begin{equation}\label{f}
f(\mathbf X):= f(\x_1):= \frac{\rho\big(|\x_1|\big)}{1+ \epsilon\rho\big(|\x_1|\big)}.
\end{equation}
Note that for any $\epsilon> 0$
\begin{equation}\label{gradient boundedness of f}
\|\nabla f\|_{L_\infty}< 1.
\end{equation}

Since $\phi\in L_2(\mathbb{R}^{3N}, \mathbb{C}^{4^N})$, for $n= 1$ \eqref{simpler statement} is equivalent to
\begin{equation}\label{modified statement}
\|\rme^{af}\phi\|_{L_2(\mathbb{R}^{3N}, \mathbb{C}^{4^N})}\leqslant C
\end{equation}
with $C$ independent of $\epsilon$. Note that for any $\epsilon> 0$ the function $\rme^{af}$ is twice differentiable with bounded derivatives. Hence multiplication by $\rme^{af}$ is a bounded operator in the Sobolev spaces $H^s(\mathbb R^3, \mathbb C^4)$ with $s\in [0, 2]$.

The following two lemmata will be important in the subsequent proof.

\begin{lemma}\label{main lemma}
For any $a_0\in [0, 1)$ there exists $C(a_0)> 0$ such that for any $a\in [0, a_0]$ and $\psi\in L_2(\mathbb{R}^3, \mathbb{C}^4)$
\begin{equation}\label{main estimate}
\big\|[\Lambda_1, \rme^{af}]\psi\big\|_{H^1(\mathbb{R}^3, \mathbb{C}^4)}\leqslant C(a_0)a\|\rme^{af}\psi\|,
\end{equation}
and
\begin{equation}\label{variation of main estimate}
\big\|\rme^{-af}[\Lambda_1, \rme^{af}]\psi\big\|_{H^1(\mathbb{R}^3, \mathbb{C}^4)}\leqslant C(a_0)a\|\psi\|.
\end{equation}
\end{lemma}
Lemma~\ref{main lemma} is proved in Section~\ref{proof of the main lemma}. Some analogous estimates with $L_2$--norms instead of $H^1$--norms can be found in \cite{BachMatte2001}.

\begin{corollary}\label{auxiliary corollary}
For any $a_0\in[0, 1)$ there exists $C(a_0)> 0$ such that for any $a\in [0, a_0]$ and $\psi\in L_2(\mathbb{R}^3, \mathbb{C}^4)$
\begin{equation}\label{desired boundedness}
\|\rme^{-af}\Lambda_1\rme^{af}\psi\|\leqslant C(a_0)\|\psi\|.
\end{equation}
\end{corollary}
\begin{proof}{}
\begin{equation*}
\rme^{-af}\Lambda_1\rme^{af}= \Lambda_1+ \rme^{-af}[\Lambda_1, \rme^{af}],
\end{equation*}
and \eqref{variation of main estimate} implies \eqref{desired boundedness}.
\end{proof}

\begin{lemma}\label{projection in ball lemma}
Let $B_R$ be the ball of radius $R> 0$ in $\mathbb R^3$ centred at the origin. For any $a\in [0, 1/2)$ there exist $C(R)> 0$ and $C(a, R)> 0$ such that for any $\psi\in H^{1/2}(\mathbb R^3, \mathbb C^4)$
\begin{equation}\label{projection in ball estimate}
\|\Lambda_1\psi\|_{H^{1/2}(B_R, \mathbb C^4)}\leqslant C(R)\|\psi\|_{H^{1/2}(B_{3R}, \mathbb C^4)}+ C(a, R)\|\rme^{-2af}\psi\|_{L_2(\mathbb R^3, \mathbb C^4)}.
\end{equation}
\end{lemma}
We prove Lemma~\ref{projection in ball lemma} in Section~\ref{proof of projection in ball section}.

In order to be able to apply Lemma~\ref{projection in ball lemma} we will only consider $a\in [0, 1/2)$. We can thus fix $a_0\in [1/2, 1)$ and no longer trace the dependence of the constants in Lemma~\ref{main lemma} and Corollary~\ref{auxiliary corollary} on this parameter.

Let us fix a cluster decomposition
\begin{equation}\label{Z_0}
Z_0:= \big(\{2, \dots, N\}, \{1\}\big).
\end{equation}
Then
\begin{equation}\label{projector summation}
\Lambda_1\rme^{af}\phi= P^TP^E\Lambda_1\rme^{af}\phi= \sum_{(T_1, E_1; T_2, 1)\underset{Z_0}\prec (T, E)}(P^{T_1}P^{E_1}\otimes P^{T_2})\Lambda_1\rme^{af}\phi.
\end{equation}
The eigenfunction $\phi$ belongs to the form domain of $\mathcal{H}_N^{T, E}$, which is \begin{equation*}P^TP^E\underset{n= 1}{\overset N\otimes}\Lambda_nH^{1/2}(\mathbb R^3, \mathbb C^4)\subset H^{1/2}(\mathbb R^{3N}, \mathbb C^{4^N}).\end{equation*}
Hence by \eqref{projector summation}, \eqref{E}, \eqref{bottom}, and \eqref{E(T)}
\begin{equation}\label{varkappa from below}\begin{split}
&\langle\Lambda_1\rme^{af}\phi, (\mathcal H_{Z_0, 1}+ \mathcal H_{Z_0, 2})\Lambda_1\rme^{af}\phi\rangle\\ &\geqslant \langle\Lambda_1\rme^{af}\phi, \sum_{(T_1, E_1; T_2, 1)\underset{Z_0}\prec (T, E)}\big(\varkappa_1(Z_0, T_1, E_1)+ \varkappa_2(Z_0, T_2, 1)\big)(P^{T_1}P^{E_1}\otimes P^{T_2})\Lambda_1\rme^{af}\phi\rangle\\ &\geqslant \varkappa(T, E)\|\Lambda_1\rme^{af}\phi\|^2.
\end{split}\end{equation}

Let us introduce
\begin{eqnarray}\label{Q_1}
Q_1&:=& \varkappa(T, E)\langle \rme^{af}\phi, [\rme^{af}, \Lambda_1]\phi\rangle,\\
\label{Q_2}
Q_2&:=& \langle\Lambda_1\rme^{af}\phi, \Big(\sum_{n= 2}^N(D_n+ V_n)+ \sum_{1< n< j}^NU_{nj}+ D_1\Big)[\Lambda_1, \rme^{af}]\phi\rangle,\\
\label{Q_3}
Q_3&:=& \langle\Lambda_1\rme^{af}\phi, [D_1, \rme^{af}]\phi\rangle,\\
\label{Q_4}
Q_4&:=& -\langle\Lambda_1\rme^{af}\Lambda_1\rme^{af}\phi, \Big(V_1+ \sum_{j= 2}^NU_{1j}\Big)\phi\rangle.
\end{eqnarray}
Then by \eqref{varkappa from below} (recall the definitions \eqref{H_Z_1}, \eqref{H_Z_2}, \eqref{reduced spaces}, and \eqref{H_N'})
\begin{equation}\label{main chain}\begin{split}
&\varkappa(T, E)\|\rme^{af}\phi\|^2= \langle\Lambda_1\rme^{af}\phi, \varkappa(T, E)\Lambda_1\rme^{af}\phi\rangle+ Q_1\\ &\leqslant \langle\Lambda_1\rme^{af}\phi, (\mathcal H_{Z_0, 1}+ \mathcal H_{Z_0, 2})\Lambda_1\rme^{af}\phi\rangle+ Q_1\\ &=\langle\Lambda_1\rme^{af}\phi, \Big(\sum_{n= 2}^N(D_n+ V_n)+ \sum_{1< n< j}^NU_{nj}+ D_1\Big)\rme^{af}\phi\rangle+ Q_1+ Q_2\\&= \langle\Lambda_1\rme^{af}\phi, \rme^{af}\Big(\sum_{n= 2}^N(D_n+ V_n)+ \sum_{1< n< j}^NU_{nj}+ D_1\Big)\phi\rangle+ \sum_{l= 1}^3Q_l\\ &= \langle\Lambda_1\rme^{af}\phi, \rme^{af}\mathcal H_N^{T, E}\phi\rangle+ \sum_{l= 1}^4Q_l= \lambda\|\Lambda_1\rme^{af}\phi\|^2+ \sum_{l= 1}^4Q_l\\ &\leqslant \lambda\|\rme^{af}\phi\|^2+ \sum_{l= 1}^4Q_l.
\end{split}\end{equation}
Thus
\begin{equation}\label{the main estimate}
\big(\varkappa(T, E)- \lambda\big)\|\rme^{af}\phi\|^2\leqslant \sum_{l= 1}^4Q_l,
\end{equation}
and it remains to estimate $Q_1, \dots, Q_4$. This will be done in the next four lemmata.

\begin{lemma}\label{Q_1 lemma}
There exists a positive constant $C_1$ such that
\begin{equation}\label{Q_1 estimate}
|Q_1|\leqslant C_1a\|\rme^{af}\phi\|^2.
\end{equation}
\end{lemma}
\begin{proof}{}
By \eqref{Q_1} and Lemma~\ref{main lemma} we have
\begin{equation*}\label{Q_1 calculation}
|Q_1|\leqslant \big|\varkappa(T, E)\big|\|\rme^{af}\phi\|\big\|[\rme^{af}, \Lambda_1]\phi\big\|\leqslant Ca\big|\varkappa(T, E)\big|\|\rme^{af}\phi\|^2.
\end{equation*}
\end{proof}

\begin{lemma}\label{Q_2 lemma}
There exists a positive constant $C_2$ such that
\begin{equation}\label{Q_2 estimate}
|Q_2|\leqslant C_2a\|\rme^{af}\phi\|^2.
\end{equation}
\end{lemma}
\begin{proof}{}
Since $\Lambda_1$ commutes with $\sum_{n= 2}^N(D_n+ V_n)+ \sum_{1< n< j}^NU_{nj}$, $\phi= \Lambda_1\phi$, and $\Lambda_1[\Lambda_1, \rme^{af}]\Lambda_1= 0$, we have
\begin{equation}\label{zero}
\langle\Lambda_1\rme^{af}\phi, \bigg(\sum_{n= 2}^N(D_n+ V_n)+ \sum_{1< n< j}^NU_{nj}\bigg)[\Lambda_1, \rme^{af}]\phi\rangle= 0.
\end{equation}
According to Lemma~\ref{main lemma}
\begin{equation*}\label{little chain}
\big|\langle\Lambda_1\rme^{af}\phi, D_1[\Lambda_1, \rme^{af}]\phi\rangle\big|\leqslant \|\Lambda_1\rme^{af}\phi\|\big\||D_1|[\Lambda_1, \rme^{af}]\phi\big\|\leqslant Ca\|\rme^{af}\phi\|^2.
\end{equation*}
By \eqref{Q_2} and \eqref{zero} this implies \eqref{Q_2 estimate}.
\end{proof}

\begin{lemma}\label{Q_3 lemma}
There exists a positive constant $C_3$ such that
\begin{equation}\label{Q_3 estimate}
|Q_3|\leqslant C_3a\|\rme^{af}\phi\|^2.
\end{equation}
\end{lemma}
\begin{proof}{}
We have $[D_1, \rme^{af}]= [-i\al\cdot\nabla, \rme^{af}]= -i\al\cdot(\nabla \rme^{af})= -i\al\cdot a(\nabla f)\rme^{af}$.
Now \eqref{Q_3 estimate} follows from \eqref{Q_3} and \eqref{gradient boundedness of f}.
\end{proof}

\begin{lemma}\label{Q_4 lemma}
There exist $C_4> 0$ and $C_0(a)> 0$ such that
\begin{equation}\label{Q_4 estimate}
Q_4\leqslant C_4a\|\rme^{af}\phi\|^2+ C_0(a)\|\phi\|_{H^{1/2}}^2.
\end{equation}
\end{lemma}
\noindent We give a proof of Lemma~\ref{Q_4 lemma} in Section~\ref{Q_4 proof section}.

Substituting the estimates \eqref{Q_1 estimate}, \eqref{Q_2 estimate}, \eqref{Q_3 estimate}, and \eqref{Q_4 estimate} into \eqref{the main estimate}, we conclude that
\begin{equation}\label{main revisited}
\Big(\varkappa(T, E)- \lambda- a\sum_{l= 1}^4C_l\Big)\|\rme^{af}\phi\|^2\leqslant C_0(a)\|\phi\|_{H^{1/2}}^2.
\end{equation}
Now if
\begin{equation*}
a< \min\bigg\{\frac12, \Big(\sum_{l= 1}^4C_l\Big)^{-1}\big(\varkappa(T, E)- \lambda\big)\bigg\},
\end{equation*}
then the expression in brackets on the \lhs of \eqref{main revisited} is positive, and \eqref{main revisited} implies \eqref{modified statement} with a finite $C$ independent of $\epsilon$. Theorem~\ref{decay theorem} is proved.

\section{Boundedness of integral operators}\label{integral operators section}

In this section we collect some auxiliary material for the subsequent proofs of Lemmata~\ref{main lemma}, \ref{projection in ball lemma}, and \ref{Q_4 lemma}.
In order to be able to obtain the information on the boundedness of (singular) integral operators we will need the following two theorems:
\begin{theorem}\label{Stein theorem}
{\bfseries\em (Stein \cite{Stein1970}, Chapter~2, Section~3.2)} Let $K:\mathbb{R}^n\rightarrow\mathbb{C}$ be a measurable function such that for some $B> 0$
\begin{equation*}\label{condition 1}
\big|K(\x)\big|\leqslant B|\x|^{-n},\quad \big|\nabla K(\x)\big|\leqslant B|\x|^{-n- 1},\quad \textrm{for almost every}\quad \x\in \mathbb R^n, 
\end{equation*}
and
\begin{equation*}\label{condition 3}
\int_{R_1< |\x|< R_2}K(\x)\rmd^n\x= 0, \quad\text{for all}\quad 0< R_1< R_2< \infty.
\end{equation*}
For $g\in L_p(\mathbb{R}^n),\: 1< p< \infty$, let
\begin{equation*}\label{convolution}
A_\varepsilon(g)(\x):= \int_{|\x- \y|\geqslant \varepsilon}K(\x- \y)g(\y)\rmd^n\y,\quad \varepsilon> 0.
\end{equation*}
Then
\begin{equation}\label{boundedness}
\big\|A_\varepsilon(g)\big\|_p\leqslant B_p\|g\|_p
\end{equation}
with $B_p$ independent of $g$ and $\varepsilon$.
\end{theorem}
\begin{remark}
Inequality \eqref{boundedness} shows that the operator $A:= \underset{\varepsilon\rightarrow +0}{\lim}A_\varepsilon$ exists as a bounded operator in $L_p(\mathbb{R}^n)$ and its norm satisfies $\|A\|_p\leqslant B_p$.
\end{remark}
The second theorem is known as Schur's test:
\begin{theorem}\label{Schur test theorem}
Let $(\Omega_1, \mu_1)$ and $(\Omega_2, \mu_2)$ be two spaces with measures.
Let $A(\cdot, \cdot)$ be a measurable (matrix) function on $\Omega_1\times\Omega_2$ satisfying
\begin{equation*}\label{Schur test}
M_1:= \sup_{\y\in\Omega_2}\int_{\Omega_1}\big|A(\x, \y)\big|\rmd\mu_1(\x)< \infty, \quad M_2:= \sup_{\x\in\Omega_1}\int_{\Omega_2}\big|A(\x, \y)\big|\rmd\mu_2(\y)< \infty.
\end{equation*}
Then the integral operator
\begin{equation*}\label{integral operator}
(A\psi)(\x):= \int_{\Omega_2}A(\x, \y)\psi(\y)\rmd\mu_2(\y)
\end{equation*}
is bounded from $L_2(\Omega_2)$ to $L_2(\Omega_1)$ and $\|A\|\leqslant \sqrt{M_1M_2}$.
\end{theorem}
We will only use Theorem~\ref{Schur test theorem} in the case $\Omega_1= \Omega_2= \mathbb R^3$ with Lebesgue measure.

Note that in the case of convolution (i.e. for $A(\x, \y)= A(\x- \y)$, $\Omega_1= \Omega_2= \mathbb R^d$) Theorem~\ref{Schur test theorem} reduces to Young's inequality for convolution with $L_1$--function (see e. g. \cite{ReedSimon1975}).

For a $4\times 4$ measurable matrix function $A$ on $\mathbb{R}^3\times\mathbb{R}^3$ we define the corresponding integral operator by
\begin{equation}\label{int to kernel}
(Ag)(\x):= \underset{\varepsilon\rightarrow +0}{\lim}\int_{|\x- \y|> \varepsilon}A(\x, \y)g(\y)\rmd\y, \quad g\in C^1_0(\mathbb{R}^3, \mathbb{C}^4).
\end{equation}
We will only work with such $A$ for which \eqref{int to kernel} is well defined and extends to a bounded operator in $L_2(\mathbb{R}^3, \mathbb{C}^4)$ either by Theorem~\ref{Stein theorem} (in which case $A(\x, \y)$ has to depend only on $(\x- \y)$), or by Theorem~\ref{Schur test theorem}.

In particular, according to the definition given above and Appendix B of \cite{MorozovVugalter2006}, the integral kernel of $(\Lambda_m- 1/2)$ is
\begin{equation}\begin{split}\label{the integral kernel of Lambda}
\mathcal K(\x, \y)&= \mathcal K(\x- \y):= \frac{im}{2\pi^2}\frac{\al\cdot(\x- \y)}{\rxy^3}K_1\big(m\rxy\big)\\ &+
\frac{m^2}{4\pi^2}\Big(\beta\frac{K_1\big(m|\x- \y|\big)}{|\x- \y|}+ \frac{i\al\cdot(\mathbf{x}- \mathbf{y})}{|\x- \y|^2}K_0\big(m\rxy\big)\Big).
\end{split}\end{equation}
The boundedness follows from Theorem~\ref{Stein theorem} and \eqref{derivatives}.

Note that the function \eqref{the integral kernel of Lambda} rapidly decays together with its derivatives if $|\x- \y|$ becomes big. Namely, if for $r> 0$ we define
\begin{equation}\label{G}
G(r):= \sup_{|\x- \y|> r}\big|\mathcal K(\x, \y)\big|+ \sup_{|\x- \y|> r}\big|\nabla_\x\mathcal K(\x, \y)\big|,
\end{equation}
then by \eqref{derivatives} and the first asymptotic in \eqref{asymptotics}, for any $R> 0$ there exists $C(R)> 0$ such that
\begin{equation}\label{decay of G}
G(r)\leqslant C(R)r^{-3/2}\rme^{-r}, \quad  \textrm{for all}\quad r\geqslant R.
\end{equation}

We will also use the following elementary lemma (Lemma~10 of \cite{Morozov2008}):
\begin{lemma}\label{multiplicator lemma}
For any $d, k\in\mathbb{N}$ there exists $C> 0$ such that for any bounded differentiable function $\chi$ on $\mathbb R^d$ with bounded gradient and $u\in H^{1/2}(\mathbb{R}^d, \mathbb{C}^k)$
\begin{equation*}\label{multiplier}
\|\chi u\|_{H^{1/2}(\mathbb{R}^d, \mathbb{C}^k)}\leqslant C\big(\|\chi\|_{L_\infty(\mathbb{R}^d)}+ \|\nabla\chi\|_{L_\infty(\mathbb{R}^d)}\big)\|u\|_{H^{1/2}(\mathbb{R}^d, \mathbb{C}^k)}.
\end{equation*}
\end{lemma}

\section{Proof of Lemma~\ref{main lemma}}\label{proof of the main lemma}
\markright{\hfill\thesection\quad Proof of Lemma~\ref{main lemma}\hfill}

To prove \eqref{main estimate} it is enough to show that $[\Lambda_1, \rme^{af}]\rme^{-af}$ is a bounded operator from $L_2(\mathbb R^3, \mathbb C^4)$ to $H^1(\mathbb R^3, \mathbb C^4)$ satisfying
\begin{equation}\label{estimate for the second term}
\big\|[\Lambda_1, \rme^{af}]\rme^{-{af}}\big\|_{L_2(\mathbb{R}^3, \mathbb{C}^4)\rightarrow H^1(\mathbb{R}^3, \mathbb{C}^4)}\leqslant C(a_0)a, \quad a\in [0, 1).
\end{equation}
The integral kernel of $[\Lambda_1, \rme^{af}]\rme^{-{af}}= \big[(\Lambda_1- 1/2), \rme^{af}]\rme^{-{af}}$ is given by (see \eqref{the integral kernel of Lambda})
\begin{equation}\label{integral kernel 1}\begin{split}
\big([\Lambda_1, \rme^{af}]\rme^{-{af}}\big)(\x, \y)&= \mathcal K(\x, \y)(1- \rme^{a(f(\x)- f(\y))}),
\end{split}\end{equation}
and its gradient in $\x$ is 
\begin{equation}\begin{split}\label{integral kernel 2}
&\big(\nabla[\Lambda_1, \rme^{af}]\rme^{-{af}}\big)(\x, \y)= (\nabla_\x\mathcal K)(\x, \y)(1- \rme^{a(f(\x)- f(\y))})\\ &+ a\mathcal K(\x, \y)(1- \rme^{a(f(\x)- f(\y))})(\nabla f)(\x)- a\mathcal K(\x, \y)(\nabla f)(\x).
\end{split}\end{equation}
We rewrite
\begin{equation}\label{as Taylor}
1- \rme^{a(f(\x)- f(\y))}= -a(\nabla f)(\y)\cdot(\x- \y)+ R_1(\x, \y)+ R_2(\x, \y),
\end{equation}
where
\begin{equation*}\label{R_1}
R_1(\x, \y):= 1+ a\big(f(\x)- f(\y)\big)- \rme^{a(f(\x)- f(\y))},
\end{equation*}
and
\begin{equation*}\label{R_2}
R_2(\x, \y):= a\big((\nabla f)(\y)\cdot(\x- \y)+ f(\y)- f(\x)\big).
\end{equation*}
Since
\begin{equation*}
|\rme^z- 1- z|\leqslant (\rme- 2)z^2 \quad \textrm{for}\quad |z|\leqslant 1,
\end{equation*}
by \eqref{gradient boundedness of f} we have
\begin{equation}\label{R_1 inside}
\big|R_1(\x, \y)\big|\leqslant (\rme- 2)a^2\big(f(\x)- f(\y)\big)^2\leqslant (\rme- 2)a^2|\x- \y|^2, \ \textrm{for}\ |\x- \y|\leqslant a^{-\frac12}.
\end{equation}
On the other hand, since $a< a_0< 1$, for $|\x- \y|> a^{-\frac12}$ the functions
\begin{equation*}
\big|\mathcal K(\x, \y)R_1(\x, \y)\big| \quad \textrm{and}\quad \big|\nabla_\x\mathcal K(\x, \y)R_1(\x, \y)\big|
\end{equation*}
are integrable in $\x$ or $\y$ with the integrals bounded by $C(a_0)a$, as follows from \eqref{G}, \eqref{decay of G}, and \eqref{gradient boundedness of f}. Since $f\in C^2(\mathbb R^3)$, by the Taylor formula we have
\begin{equation*}
f(\x)- f(\y)= (\nabla f)(\y)\cdot(\x- \y)+ \langle(\mathcal Df)\big(\xi\x+ (1- \xi)\y\big)(\x- \y),(\x- \y)\rangle_{\mathbb R^3},
\end{equation*}
where $\mathcal Df$ is the Hessian matrix (i. e. the matrix of the second partial derivatives of $f$) and $\xi\in [0, 1]$. Hence
\begin{equation}\label{R_2 estimate}
\big|R_2(\x, \y)\big|= a\Big|\langle(\mathcal Df)\big(\xi\x+ (1- \xi)\y\big)(\x- \y),(\x- \y)\rangle_{\mathbb R^3}\Big|\leqslant a\|\mathcal Df\|_{L_\infty}|\x- \y|^2,
\end{equation}
where $\|\mathcal Df\|_{L_\infty}$ is bounded uniformly in $\epsilon$ by \eqref{f} and \eqref{rho}.
Substituting \eqref{as Taylor} into \eqref{integral kernel 1} and \eqref{integral kernel 2}, and using the estimates \eqref{R_1 inside} --- \eqref{R_2 estimate} we obtain \eqref{estimate for the second term} by Theorems \ref{Stein theorem} and \ref{Schur test theorem}. This completes the proof of \eqref{main estimate}.

The proof of \eqref{variation of main estimate} is completely analogous since the integral kernel of
\begin{equation*}
\rme^{-af}[\Lambda_1, \rme^{af}]= \rme^{-af}\big[(\Lambda_1- 1/2), \rme^{af}\big]
\end{equation*}
is 
\begin{equation*}
\mathcal K(\x, \y)(\rme^{a(f(\y)- f(\x))}- 1)
\end{equation*}
(compare with \eqref{integral kernel 1}).

\section{Proof of Lemma~\ref{projection in ball lemma}}\label{proof of projection in ball section}
\markright{\hfill\thesection\quad Proof of Lemma~\ref{projection in ball lemma}\hfill}

Let $\eta\in C^\infty\big(\mathbb R^3, [0,1]\big)$ with
\begin{equation*}
\eta(\x)\equiv \begin{cases}0, &\x\in B_{2R},\\ 1, & \x\in \mathbb R^3\setminus B_{3R}.\end{cases}
\end{equation*}
Since $\Lambda_1$ is a bounded operator in $H^{1/2}(\mathbb R^3, \mathbb C^4)$, by Lemma~\ref{multiplicator lemma} we have
\begin{equation}\label{obvious thing}\begin{split}
\|\Lambda_1\psi\|_{H^{1/2}(B_R, \mathbb C^4)}&\leqslant \big\|\Lambda_1(1- \eta)\psi\big\|_{H^{1/2}(B_{R}, \mathbb C^4)}+ \|\Lambda_1\eta\psi\|_{H^{1/2}(B_R, \mathbb C^4)}\\ &\leqslant C(R)\|\psi\|_{H^{1/2}(B_{3R}, \mathbb C^4)}+ \|\Lambda_1\eta\psi\|_{H^1(B_R, \mathbb C^4)}.
\end{split}\end{equation}
By \eqref{G} we can estimate the second term on the \rhs of \eqref{obvious thing} as
\begin{equation*}\label{with G}\begin{split}
&\|\Lambda_1\eta\psi\|_{H^1(B_R, \mathbb C^4)}^2\\ &= \int\limits_{B_R}\bigg(\Big|\int\limits_{|\y|> 2R}K(\x, \y)\eta(\y)\psi(\y)\rmd\y\Big|^2+ \Big|\int\limits_{|\y|> 2R}\nabla_\x K(\x, \y)\eta(\y)\psi(\y)\rmd\y\Big|^2\bigg)\rmd\x\\ &\leqslant \frac43\pi R^3\sup_{\x\in B_R}\bigg(\Big|\int_{|\y|> 2R}K(\x, \y)\eta(\y)\psi(\y)\rmd\y\Big|^2\\ &+ \Big|\int_{|\y|> 2R}\nabla_\x K(\x, \y)\eta(\y)\psi(\y)\rmd\y\Big|^2\bigg)\\ &\leqslant \frac43\pi R^3\bigg(\int_{|\y|> 2R}\Big(\sup_{\x\in B_R}\big|K(\x, \y)\big|+ \sup_{\x\in B_R}\big|\nabla_\x K(\x, \y)\big|\Big)\big|\psi(\y)\big|\rmd\y\bigg)^2\\ & \leqslant \frac43\pi R^3\bigg(\int_{|\y|> 2R}G\big(|\y|- R\big)\big|\psi(\y)\big|\rmd\y\bigg)^2\\ &\leqslant \frac43\pi R^3\bigg(\int_{|\y|> 2R}G^{1- 2a}\big(|\y|- R\big)\rmd\y\bigg)\bigg(\int_{|\y|> 2R}G^{1+ 2a}\big(|\y|- R\big)\big|\psi(\y)\big|^2\rmd\y\bigg).
\end{split}\end{equation*}
Since $a< 1/2$ and $f(\x)\leqslant |\x|$, we conclude from \eqref{decay of G} that there exists $C(a, R)$ such that
\begin{equation*}
\|\Lambda_1\eta\psi\|_{H^1(B_R, \mathbb C^4)}\leqslant C(a, R)\|\rme^{-2af}\psi\|_{L_2(\mathbb R^3, \mathbb C^4)},
\end{equation*}
and \eqref{projection in ball estimate} follows by \eqref{obvious thing}.

\section{Proof of Lemma~\ref{Q_4 lemma}}\label{Q_4 proof section}
\markright{\hfill\thesection\quad Proof of Lemma~\ref{Q_4 lemma}\hfill}

For $j= 2, \dots, N$ we have
\begin{equation}\begin{split}\label{U_1j calculation}
&\langle\Lambda_1\rme^{af}\Lambda_1\rme^{af}\phi, U_{1j}\phi\rangle= \langle U_{1j}\rme^{af}\phi, \rme^{af}\phi\rangle\\ &+\langle U_{1j}\rme^{-af}[\Lambda_1, \rme^{af}]\Lambda_1\rme^{af}\phi, \rme^{af}\phi\rangle+ \langle U_{1j}[\Lambda_1, \rme^{af}]\phi, \rme^{af}\phi\rangle.
\end{split}\end{equation}
The first term on the \rhs of \eqref{U_1j calculation} is nonnegative by \eqref{U positivity}. Applying \eqref{U Hardy}, Lemma~\ref{main lemma}, and Schwarz inequality we can estimate the last two terms by $Ca\|\rme^{af}\phi\|^2$. Hence by \eqref{Q_4}
\begin{equation}\label{V_1 remains}
Q_4\leqslant Ca\|\rme^{af}\phi\|^2+ \big|\langle\Lambda_1\rme^{af}\Lambda_1\rme^{af}\phi, V_1\phi\rangle\big|
\end{equation}
and it remains to estimate the last term on the \rhs of \eqref{V_1 remains}.

Let $\chi_1\in C^\infty\big(\mathbb{R}^3, [0, 1]\big)$ be a function supported in $\mathbb{R}^3\setminus B_1$ such that it is equal to $1$ on $\mathbb{R}^3\setminus B_2$. For $R> 1$ let
\begin{equation*}\label{chi_R}
\chi_R(\mathbf{X}):= \chi_R(\x_1):= \chi_1(\x_1/R).
\end{equation*}
We have
\begin{equation}\label{Bs}\begin{split}
\big|\langle\Lambda_1\rme^{af}\Lambda_1\rme^{af}\phi, V_1\phi\rangle\big|&\leqslant \big|\langle \rme^{-af}\Lambda_1\rme^{af}\Lambda_1\rme^{af}\phi, \chi_RV_1\rme^{af}\phi\rangle\big|\\ &+ \big|\langle(1- \chi_R)\Lambda_1\rme^{af}\Lambda_1\rme^{af}\phi, V_1\phi\rangle\big|.
\end{split}\end{equation}
By Corollary~\ref{auxiliary corollary},
\begin{equation}\label{with exponents around}
\|\rme^{-af}\Lambda_1\rme^{af}\Lambda_1\rme^{af}\phi\|\leqslant C\|\rme^{af}\phi\|.
\end{equation}
Since $\chi_R$ is supported outside $B_R$, by \eqref{V_n decay again} we have
\begin{equation}\label{for V_1}
\|\chi_RV_1\rme^{af}\phi\|\leqslant \varepsilon(R)\|\rme^{af}\phi\|, \quad \varepsilon(R)\underset{R\to \infty}\longrightarrow 0.
\end{equation}
According to \eqref{upper estimate V},
\begin{equation}\label{by form boundedness}
\big|\langle(1- \chi_R)\Lambda_1\rme^{af}\Lambda_1\rme^{af}\phi, V_1\phi\rangle\big|\leqslant C\big\|(1- \chi_R)\Lambda_1\rme^{af}\Lambda_1\rme^{af}\phi\big\|_{H_1^{1/2}}\|\phi\|_{H_1^{1/2}}.
\end{equation}
Since $(1- \chi_R)$ is a smooth function supported in $\big\{|\x_1|\leqslant 2R\big\}$, by Lemmata~\ref{multiplicator lemma} and \ref{projection in ball lemma} we have
\begin{equation}\label{with lemma}\begin{split}
&\big\|(1- \chi_R)\Lambda_1\rme^{af}\Lambda_1\rme^{af}\phi\big\|_{H_1^{1/2}}\leqslant C(R)\|\Lambda_1\rme^{af}\Lambda_1\rme^{af}\phi\|_{H_1^{1/2}(B_{2R}\times\mathbb R^{3N- 3}, \mathbb C^{4^N})}\\ &\leqslant C(R)\|\rme^{af}\Lambda_1\rme^{af}\phi\|_{H_1^{1/2}(B_{6R}\times\mathbb R^{3N- 3}, \mathbb C^{4^N})}+ C(a, R)\|\rme^{-af}\Lambda_1\rme^{af}\phi\|_{L_2(\mathbb R^{3N}, \mathbb C^{4^N})}.
\end{split}\end{equation}
By Corollary~\ref{auxiliary corollary} the second term on the \rhs of \eqref{with lemma} can be estimated by $C(a, R)\|\phi\|$. Applying Lemma~\ref{projection in ball lemma} to the first term we obtain
\begin{equation}\label{again with lemma}\begin{split}
&C(R)\|\rme^{af}\Lambda_1\rme^{af}\phi\|_{H_1^{1/2}(B_{6R}\times\mathbb R^{3N- 3}, \mathbb C^{4^N})}\\ &\leqslant C(a, R)\|\Lambda_1\rme^{af}\phi\|_{H_1^{1/2}(B_{6R}\times\mathbb R^{3N- 3}, \mathbb C^{4^N})}\\ &\leqslant C(a, R)\|\rme^{af}\phi\|_{H_1^{1/2}(B_{18R}\times\mathbb R^{3N- 3}, \mathbb C^{4^N})}+ C(a, R)\|\rme^{-af}\phi\|_{L_2(\mathbb R^{3N}, \mathbb C^{4^N})}\\ &\leqslant C(a, R)\|\phi\|_{H^{1/2}(\mathbb R^{3N}, \mathbb C^{4^N})}.
\end{split}\end{equation}
Thus by \eqref{by form boundedness} --- \eqref{again with lemma}
\begin{equation}\label{B2 finished}
\big|\langle(1- \chi_R)\Lambda_1\rme^{af}\Lambda_1\rme^{af}\phi, V_1\phi\rangle\big|\leqslant C(a, R)\|\phi\|_{H^{1/2}}^2.
\end{equation}
Estimating the \rhs of \eqref{Bs} according to \eqref{with exponents around}, \eqref{for V_1}, and \eqref{B2 finished} and substituting the result into \eqref{V_1 remains} we obtain
\begin{equation*}\label{final good thing}
Q_4\leqslant Ca\|\rme^{af}\phi\|^2+ C\varepsilon(R)\|\rme^{af}\phi\|^2+ C(a, R)\|\phi\|_{H^{1/2}}^2.
\end{equation*}
Choosing $R$ so that $\varepsilon(R)\leqslant a$ we arrive at \eqref{Q_4 estimate}.
Lemma~\ref{Q_4 lemma} is proved.

\ack

The author was supported by the DFG grant SI 348/12--2 while working on his PhD (on a part of which this article is based), and by the EPSRC grant EP/F029721/1 while preparing this publication. Many thanks to S. Vugalter for numerous valuable discussions and suggestions.

\appendix

\section{Some properties of modified Bessel functions}\label{K-functions appendix}

The modified Bessel (McDonald) functions are related to the Hankel functions by the formula
\begin{equation*}
K_\nu(z)= \frac{\pi}{2}\rme^{i\pi(\nu+ 1)/2}H_\nu^{(1)}(iz).
\end{equation*}
These functions are positive and decreasing for $z\in (0, \infty)$. Their asymptotics are (see \cite{GradshteynRyzhik2000} 8.446, 8.447.3, 8.451.6)
\begin{equation}\label{asymptotics}\begin{split}
K_\nu(z)&= \sqrt\frac{\pi}{2z}\rme^{-z}\bigg(1+ \Or\Big(\frac{1}{z}\Big)\bigg), \quad z\rightarrow+\infty;\\
K_0(z)&= -\log z\big(1+ \mathrm o(1)\big), \quad K_1(z)= \frac{1}{z}\big(1+ \mathrm o(1)\big), \quad z\rightarrow +0.
\end{split}\end{equation}
The derivatives of these functions are (see \cite{GradshteynRyzhik2000} 8.486.12, 8.486.18)
\begin{equation}\label{derivatives}
K_0'(z)= -K_1(z), \quad K_1'(z)= -K_0(z)- \frac{1}{z}K_1(z), \quad z\in (0, \infty).
\end{equation}

\section*{References}
\bibliographystyle{unsrt}
\bibliography{Bericht}

\end{document}